\documentclass[11pt]{article}
\usepackage{fullpage}
\usepackage{amssymb}
\usepackage{amsmath}
\usepackage{amsthm}
\usepackage{epsfig}
\usepackage{makeidx}
\usepackage{graphicx}
\usepackage{xcolor}

\newcommand{\commentout}[1]{}

\newtheorem{theorem}{Theorem}

\newtheorem{definition}{Definition}

\newfont{\ssbb}{cmssbx10 scaled\magstep2}
\newfont{\ssb}{cmssbx10 scaled\magstep1}
\newfont{\stb}{cmssbx10}
\newfont{\teniu}{cmu10}

\newcommand{\rb}[2]{\raisebox{#1 mm}[0mm][0mm]{#2}}

\newcommand{\ThreeSUM}{\textsf{3SUM}}

\newcommand{\degen}{\mathcal{\delta}}
\newcommand{\lsc}{lsc}
\newcommand{\clb}{CLB}

\newcommand{\Patrascu}{P\v{a}tra\c{s}cu}
\newcommand{\rl}{\mathcal{L}}
\newcommand{\ot}{\tilde O}
\newcommand{\op}{op}

\begin{document}
    \title{Mind the Gap}
  \author{
    Amihood Amir \thanks{Bar-Ilan University and Johns Hopkins University, \texttt{amir@cs.biu.ac.il} } \and
 Tsvi Kopelowitz\thanks{University of Michigan, \texttt{kopelot@gmail.com}} \and
 Avivit Levy\thanks{Shenkar College, \texttt{avivitlevy@shenkar.ac.il}} \and
  Seth Pettie\thanks{University of Michigan, \texttt{pettie@umich.edu}} \and
\and Ely Porat\thanks{Bar-Ilan University, \texttt{porately@cs.biu.ac.il}} \and
\and B.~Riva Shalom\thanks{Shenkar College, \texttt{rivash@shenkar.ac.il}}
  }

    \author{
    Amihood Amir\\
    {\footnotesize Bar-Ilan University}\\ {\footnotesize and Johns Hopkins University}\\
    \footnotesize \texttt{amir@cs.biu.ac.il}
    \and
    Tsvi Kopelowitz\\
    \footnotesize University of Michigan\\
    \footnotesize \texttt{kopelot@gmail.com}
    \and
    Avivit Levy\\
    \footnotesize Shenkar College\\
    \footnotesize \texttt{avivitlevy@shenkar.ac.il}
    \and
    Seth Pettie\\
    \footnotesize University of Michigan\\
    \footnotesize \texttt{pettie@umich.edu}
    \and
    Ely Porat\\
    \footnotesize Bar-Ilan University\\
    \footnotesize \texttt{porately@cs.biu.ac.il}
    \and
    B.~Riva Shalom\\
    \footnotesize Shenkar College\\
    \footnotesize \texttt{rivash@shenkar.ac.il}
  }

  \date{}
  \maketitle
\thispagestyle{empty}

\setcounter{page}{0}

\begin{abstract}
We examine the complexity of the online Dictionary Matching with One Gap Problem (DMOG) which is the following.
Preprocess a dictionary $D$ of $d$ patterns, where each pattern contains a special \emph{gap} symbol that can match any string, so that given a text that arrives online, a character at a time, we can report all of the patterns from $D$ that are suffixes of the text that has arrived so far, before the next character arrives. In more general versions the gap symbols are associated with \emph{bounds} determining the possible lengths of matching strings.

Finding efficient algorithmic solutions for (online) DMOG has proven to be a difficult algorithmic challenge. Little progress has been made on this problem and to date and no truly efficient solutions are known. However, the need for efficient practical solutions has been on a rise since the online DMOG is a bottleneck procedure in the realm of cyber-security, as many digital signatures of viruses manifest themselves as patterns with a gap. Indeed, this paper was invoked by such a challenge.

We first demonstrate that the difficulty in obtaining efficient solutions for the DMOG problem even, in the offline setting, can be traced back to the infamous \ThreeSUM{} conjecture, by showing that an efficient solution for DMOG implies that the \ThreeSUM{} conjecture is false. Interestingly, our reduction deviates from the known reduction paths that follow from \ThreeSUM{}. In particular, most reductions from \ThreeSUM{} go through the set-disjointness problem, which corresponds to the problem of preprocessing a graph to answer edge-triangles queries: given an edge, report all of the triangles containing this edge. We use a new path of reductions by considering the complementary, although structurally very different, \emph{vertex-triangles} queries: given an vertex, report all of the triangles containing this vertex. Using this new path we show a conditional lower bound of $\Omega(\degen(G_D)+\op)$ time per text character, where $G_D$ is a bipartite graph that captures the structure of $D$, $\degen(G_D)$ is the \emph{degeneracy} of this graph, and $\op$ is the output size. Moreover, we show a conditional lower bound in terms of the magnitude of gaps for the bounded case, thereby showing that some known upper bounds are optimal.

We also provide matching upper-bounds (up to sub-polynomial factors) for the vertex-triangles problem, and then extend these techniques to the online DMOG problem. In particular, we introduce algorithms whose time cost depends linearly on $\degen(G_D)$. Our algorithms make use of graph orientations, together with some additional techniques. These algorithms are of practical interest since although $\degen(G_D)$ can be as large as $\sqrt{d}$, and even larger if $G_D$ is a multi-graph, it is typically a very small constant in practice. Finally, when $\degen(G_D)$ is large we are able to obtain even more efficient solutions\commentout{ which are mostly of theoretical interest}.
\end{abstract}

\section{Introduction}\label{s:introduction}
Understanding the computational limitations of algorithmic problems can often lead to algorithms for such problems that work well in practice.
This paper, which stemmed from an industrial-acdemic connection~\cite{VerInt}, is a prime example of where proving algorithmic lower bounds lead to efficient practical algorithmic upper bounds. We focus on an aspect of Cyber-security which is a critical modern challenge. Network intrusion detection systems (NIDS) perform protocol analysis, content searching and content matching, in order to detect harmful software. Such malware may appear non-contiguously, scattered across several packets, which necessitates matching {\em gapped} patterns. A {\em gapped pattern} $P$ is one of the form
$P_1\ \{ \alpha_{1},\beta_{1} \}\ P_{2}\  \{\alpha_{2},\beta_{2}\}\ldots \  \{\alpha_{k},\beta_{k}\} \  P_{k+1}$, where each subpattern
$P_{j}$ is a string over alphabet $\Sigma$, and $\{\alpha_{j},\beta_{j}\}$ matches any substring of length at least $\alpha_j$ and at most $\beta_j$. Gapped patterns considered in NIDS systems typically have only \emph{one} gap. Though the gapped pattern matching problem arose over 20 years ago in computational biology applications~\cite{mm-93,hpfb-99} and has been revisited many times in the intervening years (e.g.~\cite{myers-92,bt-10,mpvz-05,bgvw:12,fp-08,nr-03,ZZH:10}), in this paper we study what is apparently a mild generalization of the problem that has nonetheless resisted many researcher's attempts at finding a definitive efficient solution.

The set of $d$ patterns to be detected, called a {\em dictionary}, could be quite large. While dictionary matching is well studied (see, e.g.~\cite{AC75,AmirFIPS95,BG:96,aklllr00,cgl-04}), NIDS applications motivate the {\em dictionary matching with one gap} problem, defined formally as follows.

\begin{definition}\label{d:gappedDictionaryMatch} {\em The Dictionary Matching with One Gap Problem (DMOG)}, is:\\
\begin{tabular}{ll}
Input: & A text $T$ of length $|T|$ over alphabet $\Sigma$, and a dictionary $D$ of $d$ gapped patterns\\
       & $P_1,\ldots,P_d$ over alphabet $\Sigma$ where each pattern has at most one gap.\\
  Output: & All locations in $T$ where a pattern $P_i\in D$, $1\leq i\leq d$, ends.\\
\end{tabular}
\end{definition}

In the offline DMOG problem $T$ and $D$ are presented all at once. We study the more practical {\em online} DMOG problem.
The dictionary $D$ can be preprocessed in advance, resulting in a data structure. Given this data structure the text $T$ is presented
one character at a time, and when a character arrives the subset of patterns with a match ending at this character should be reported before the next character arrives. Three cost measures are of interest: a preprocessing time, a time per character, and a time per match reported. Online DMOG is a serious bottleneck for NIDS, and has received much attention from both the industry and the academic community.

\paragraph{\textbf{Previous Work.}} Finding efficient solutions for DMOG has proven to be a difficult algorithmic challenge as, unfortunately, little progress has been obtained on this problem even though many researchers in the pattern matching community and the industry have tackled it. Table~\ref{tab} describes a summary and comparison of previous work. It illustrates that previous results are inadequate for NIDS, on the one hand, and that our upper bounds are basically optimal, on the other hand.

\begin{table}
\centering
\footnotesize{
  \begin{tabular}{|c|c|c|c|c|}
    \hline
     & Preprocessing & Total Query Time & Algorithm & Remark \\
    & Time & & Type & \\
    \hline
    \hline
    \rb{-2}{\cite{KuRu97}} & \rb{-2}{none} & \rb{-2}{$\tilde{O}(|T|+|D|)$} & \rb{-2}{online} & reports only \\
    & & &  & first occurrence \\
    \hline
    \rb{-2}{\cite{ZZH:10}} & \rb{-2}{$O(|D|)$} & \rb{-2}{$\tilde{O}(|T|+d)$} & \rb{-2}{online} & reports only \\
    & & & & first occurrence \\
    \hline
    \rb{-2}{\cite{hsss:11}} & \rb{-2}{$O(|D|)$} & \rb{-2}{$O(|T|\cdot \lsc + socc)$} & \rb{-2}{online} & reports one occurrence \\
     & & & & per pattern and location \\
    \hline
    \cite{alps:14} & $\ot(|D|)$ & $\ot(|T|(\beta -\alpha) + op)$ & offline & DMOG \\
    \hline
    \cite{HLSTTY:15} & $O(|D|)$ & $\ot(|T|(\beta^*-\alpha^*) + op)$ & offline & DMOG \\
    \hline
    \hline
    This &  &  &  & \\
    paper & \rb{2}{$O(|D|)$} & \rb{2}{$\ot(|T|\cdot\delta(G_D)\cdot\lsc +op)$} & \rb{2}{online} & \rb{2}{DMOG}\\
    \hline
     This & $O(|D|)$ & $\Omega(|T|\cdot\delta(G_D)^{1-o(1)} +op)$ & online & \rb{-2}{DMOG}\\
    paper & $O(|D|)$ & $\Omega(|T|\cdot(\beta-\alpha)^{1-o(1)} +op)$ & or offline & \\
    \hline
  \end{tabular}\\
\caption{\footnotesize Comparison of previous work and some new results. The parameters: $\lsc$ is the longest suffix chain of subpatterns in $D$, $socc$ is the number of subpatterns occurrences in $T$, $op$ is the number of  pattern occurrences in $T$, $\alpha^*$ and $\beta^*$ are the minimum left and maximum right gap borders in the non-uniformly bounded case, $\delta(G_D)$ is the degeneracy of the graph $G_D$ representing dictionary $D$.}\label{tab}
}
\end{table}

\subsection{Our Results}
The DMOG problem has several natural parameters, e.g., $|D|,d$, and the magnitude of the gap. We establish almost sharp upper and lower bounds for the cases of unbounded gaps ($\alpha=0, \beta=\infty$), uniformly bounded gaps where all patterns have the same bounds on their gap, and the most general non-uniform gaps version. We show that the complexity of DMOG actually depends on a ``hidden'' parameter that is a function of the {\em structure} of the gapped patterns. The dictionary $D$ can be represented as a graph $G_D$, which is a multi-graph in the non-uniformly bounded gaps case, where vertices correspond to first or second subpatterns and edges correspond to patterns. We use the notion of graph {\em degeneracy} $\degen(G_D)$ which is defined as follow. The degeneracy of an undirected graph $G=(V,E)$ is
$\degen(G) = \max_{U\subseteq V} \min_{u\in U} d_{G_U}(u)$, where $d_{G_U}$ is the degree of $u$ in the subgraph of $G$ induced by
$U$. In words, the degeneracy of $G$ is the largest minimum degree of any subgraph of $G$. It is straightforward to see that a non-multi graph $G$ with $m$ edges has $\degen(G) = O(\sqrt m)$, and a clique has $\degen(G) = \Theta(\sqrt m)$. The degeneracy of a multi-graph can be much higher.

\paragraph{Vertex-triangle queries.} A key component in understanding both the upper and lower bounds for DMOG is the \emph{vertex-triangles} problem, where the goal is to preprocess a graph so that given a query vertex $u$ we may list all triangles that contain $u$. The vertex-triangles problem, besides being a natural graph problem, is of particular interest here since, as will be demonstrated in Section~\ref{s:conditional-lower-bound}, it is reducible to DMOG. Our reduction demonstrates that the complexity of the DMOG problem already emerges when all patterns are of the form of two characters separated by an unbounded gap. This simplified online DMOG problem is equivalent to the following {\em Induced Subgraph} (ISG) problem. Preprocess a directed graph $G=(V,E)$ such that given a sequence of vertices online, after vertex $v_i$ all edges $(v_j,v_i)\in E$ with $j<i$ are reported. Thus, characters and gapped patterns in DMOG correspond to vertices and edges in ISG, respectively. We show that vertex-triangles queries are reducible to ISG.

This reduction serves two purposes. First, in Section~\ref{s:conditional-lower-bound} we prove a \emph{conditional lower bound} (\clb{}) for DMOG based on the \ThreeSUM{} conjecture by combining a straightforward reduction from triangle enumeration to the vertex-triangles problem with our new reduction from the vertex-triangles problem to DMOG. Our lower bound states that any online DMOG algorithm with low preprocessing and reporting costs must spend $\Omega(\degen(G_D)^{1-o(1)})$ per character, assuming the \ThreeSUM{} conjecture. Interestingly, the path for proving this \clb{} deviates from the common conceptual paradigms for proving lower bounds conditioned on the \ThreeSUM{} conjecture, and is of independent interest. We provide an overview of this method in Section~\ref{ss:LB_story} and the details appear in Section~\ref{s:conditional-lower-bound}. Moreover, our \clb{} holds for the \emph{offline} case as well, and can be rephrased in terms of other parameters. For example, in the DMOG problem with uniform gaps $\{\alpha,\beta\}$, we prove that the per character cost of scanning $T$ must be $\Omega((\beta-\alpha)^{1-o(1)})$. This gives some indication that some recent algorithms for the offline version of DMOG problem are almost optimal (\cite{alps:14,HLSTTY:15}).

Second, in Section~\ref{s:induced} we provide optimal solutions, up to subpolynomial factors, for ISG and, therefore, also for vertex-triangles queries, with $O(|E|)$ preprocessing time and $O(\degen(G)+op)$ time per each vertex, where $op$ is the size of the output due to the vertex. The connection between ISG and DMOG led us to extend the techniques used to solve ISG, combine them with additional ideas and techniques, thereby introduce several new online DMOG algorithms whose dependence on $\degen(G)$ is linear. Thus, graph degeneracy seems to capture the intrinsic complexity of the problem. On the other hand, the statement of our general algorithmic results is actually a bit more complicated as it depends on other parameters of the input, namely $\lsc$, the {\em longest suffix chain} in the dictionary, i.e., the longest sequence of dictionary subpatterns such that each is a proper suffix of the next. While the parameter $\lsc$ could theoretically be as large as $d$, in practice it is very small~\cite{VerInt}. Nevertheless, we also present algorithms that in the most dense cases reduce the dependence on $\lsc$. We provide an overview of our algorithms in Section~\ref{ss:UB_story}, and the details appear in Sections~\ref{s:induced},~\ref{s:DMOGinduced}, and~\ref{s:threshold}.

\paragraph{Lower Bounds Leading to Practical Upper Bounds.} After trying to tackle the DMOG problem from the upper bound perspective, we suspected that a lower bound could be proven, and indeed were successful in showing a connection to the \ThreeSUM{} conjecture. The \clb{} proof provides insight for the inherent difficulty in solving DMOG, but is also unfortunate news for those attempting to find efficient upper bounds.

Fortunately, after a careful examination of the reduction from \ThreeSUM{} to DMOG we noticed that the \clb{} from the \ThreeSUM{} conjecture can be phrased in terms of $\degen(G_D)$, which turns out to be a small constant in the input instances considered by NIDS. This lead to designing algorithms whose runtime can be expressed in terms of $\degen(G_D)$, and can therefore be helpful in practical settings.

\subsection{The Lower Bound Story}\label{ss:LB_story}
Solving DMOG with poly-logarithmic time bounds seems to be an extremely difficult task, leading
to the question of finding a polynomial time lower bound. Polynomial (unconditional) lower bounds for data structure problems are considered beyond the reach of current techniques. Thus, it has recently become extremely popular to prove \clb{}s based on the {\em conjectured} hardness of some problem. One of the most popular conjectures for \clb{}s is that the \ThreeSUM{} problem (given $n$ integers determine if any three sum to zero) cannot be solved in truly subquadratic time, where truly subquadratic time is $O(n^{2-\Omega(1)})$ time. This conjecture holds even if the the algorithm is allowed to use randomization (see e.g.~\cite{Patrascu10,AW14,KPP14a,GronlundP14}). In Section~\ref{s:conditional-lower-bound} we show that the infamous \ThreeSUM{} problem can be reduced to DMOG, which sheds some light on the difficulty of the DMOG problem. Interestingly, our reduction does not follow the common paradigm for proving \clb{}s based on the \ThreeSUM{} conjecture, providing a new approach for reductions from \ThreeSUM{}. This approach is of independent interest, and is described next.

\paragraph{Triangles.} \Patrascu{}~\cite{Patrascu10} showed that \ThreeSUM{} can be reduced to enumerating triangles in a tripartite graph. Kopelowitz, Pettie, and Porat~\cite{KPP14a} provided more efficient reductions, thereby showing that many known triangle enumeration algorithms (\cite{ItaiR78,CN85,BjorklundPWZ14,KPP15}) are essentially and conditionally optimal, up to subpolynomial factors. Hence, the offline version of triangle enumeration is well understood. The following two indexing versions of the triangle enumeration problem are a natural extension of the offline problem. In the \emph{edge-triangles} problem the goal is to preprocess a graph so that given a query edge $e$ all triangles containing $e$ are listed. The vertex-triangles problem is defined above. Clearly, both these versions solve the triangle enumeration problem, which immediately gives lower bounds conditioned on the \ThreeSUM{} conjecture.

It is fairly straightforward to see that the edge-triangles problem on a tripartite graph corresponds to preprocessing a family $F$ of sets over a universe $U$ in order to support set intersection queries in which given two sets $S,S'\in F$ the goal is to enumerate the elements in $S\cap S'$ (see~\cite{KPP14a}).
Indeed, the task of preprocessing $F$ to support set-intersection enumeration queries, and hence edge-triangles, is well studied~\cite{CP10,KPP15}. Furthermore, the set intersection problem has been used extensively as a tool for proving that many algorithmic problems are as hard as solving \ThreeSUM{}~\cite{Patrascu10,AW14,KPP14a}. However, the vertex-triangles problem has yet to be considered directly\footnote{The closely related problem of deciding whether a given vertex is contained by any triangle (a decision version) has been addressed~\cite{BW12}.}.

We use the vertex-triangles problem in order to show that the ISG problem is hard, and thus the simplest DMOG version of (offline) unbounded setting is \ThreeSUM{}-hard. The following theorems are simplified statements of the ones proven in Section~\ref{s:conditional-lower-bound}.

\begin{theorem}\label{thm:vertex_triangles_LB}
Assume \ThreeSUM{} requires $\Omega(n^{2-o(1)})$ expected time.
For any algorithm that solves the vertex-triangles problem on a graph $G$ with $m$ edges, if the amortized expected preprocessing time is $O(m\cdot \degen(G)^{1-\Omega(1)})$ and the amortized expected reporting time is sub-polynomial, then the amortized expected query time must be at least $\Omega((\hat d \cdot \degen(G))^{1-o(1)})$, where $\hat{d}$ is the degree of the queried vertex.
\end{theorem}

\begin{theorem}\label{thm:ISG_lb}
Assume \ThreeSUM{} requires $\Omega(n^{2-o(1)})$ expected time.
For any algorithm that solves the ISG problem on a graph $G$ with $m$ edges, if the amortized expected preprocessing time is $O(m\cdot \degen(G)^{1-\Omega(1)})$ and the amortized expected reporting time is sub-polynomial, then the amortized expected time spent on each vertex during a query must be at least $\Omega((\degen(G))^{1-o(1)})$.
\end{theorem}

\begin{theorem}\label{thm:DMOG_lb}
Assume \ThreeSUM{} requires $\Omega(n^{2-o(1)})$ expected time.
For any algorithm that solves the DMOG problem on a graph $G$ with $m$ edges, if the amortized expected preprocessing time is $O(|D|\cdot \degen(G_D)^{1-\Omega(1)})$ and the amortized expected reporting time is sub-polynomial, then the amortized expected time spent on each text character must be at least $\Omega((\degen(G_D))^{1-o(1)})$.
\end{theorem}

\begin{theorem}\label{thm:bounded_DMOG_lb}
Assume \ThreeSUM{} requires $\Omega(n^{2-o(1)})$ expected time.
For any algorithm that solves the uniformly bounded DMOG problem on a graph $G$ with $m$ edges, if the amortized expected preprocessing time is $O(|D|\cdot \degen(G_D)^{1-\Omega(1)})$ and the amortized expected reporting time is sub-polynomial, then the amortized expected time spent on each text character must be at least $\Omega((\beta-\alpha)^{1-o(1)})$.
\end{theorem}

\paragraph{A Note on Triangle Reporting Problems and Other Popular Conjectures.}
Many \clb{}s based on other popular conjectures, such as the Boolean Matrix Multiplication conjecture or the Online Matrix Vector Multiplication conjecture, use reductions from set-disjointness and hence from edge-triangles queries (see~\cite{AW14,HKNS15}). However, it is not clear how to obtain meaningful lower bounds for vertex-triangles queries based on these conjectures. These difficulties are discussed in Appendix~\ref{s:app_other_clbs}.

\subsection{The Upper Bound Story}\label{ss:UB_story}
Given Theorems~\ref{thm:vertex_triangles_LB},~\ref{thm:ISG_lb}, and~\ref{thm:DMOG_lb}, we focus on providing matching upper bounds for the online versions of the vertex-triangles problem, ISG, and DMOG. A table summarizing our upper-bounds for DMOG appears in Appendix~\ref{app:upper-bound-table}. In Section~\ref{s:induced}, we provide a solution for ISG using $O(m)$ preprocessing time and $O(\degen(G) + \op)$ query time using the graph orientation technique, which is closely related to graph degeneracy. This matches the \clb{} from Theorem~\ref{thm:ISG_lb}, and by applying the reduction from vertex-triangles queries to ISG, it also matches the \clb{} from Theorem~\ref{thm:vertex_triangles_LB}.

Our ISG algorithms are then extended to versions corresponding to the simplified versions of the uniformly and non-uniformly bounded DMOG problems. This is shown in Sections~\ref{ss:uniformly-bounded-ISG} and~\ref{ss:non-uniformly-bounded-ISG}. The time bounds, ignoring poly-log factors, remain the same, however, the space usage is slightly increased. Interestingly, for the non-uniform case we utilize 4-sided 2-dimensional orthogonal range reporting queries in a clever way.

In Section~\ref{s:DMOGinduced}, the ISG algorithms are extended to solutions for the various DMOG versions. However, the longer subpatterns introduce new challenges that need to be tackled. First, since subpatterns can be suffixes of each other, up to $\lsc$ vertices can arrive simultaneously in $G_D$, which is the backbone of our algorithms. Thus, the time of our algorithms have a multiplicative factor of $\lsc$. We emphasize that we are not the first to introduce the $\lsc$ factor in solutions for DMOG problems~\cite{hsss:11}. Second, since subpatterns may be long, we must accommodate a delay in the time a vertex corresponding to a second subpattern is treated as if it has arrived, thus inducing a minor additive space usage.

Finally, the online DMOG algorithms in Section~\ref{s:DMOGinduced} have a cost per character of $O(\lsc \cdot \degen(G_D)+\op)$. In Section~\ref{s:threshold} we show that it is possible to obtain more efficient bounds that depend linearly on $\sqrt{\lsc\cdot d}$ when $\degen(D_G) \geq \sqrt{\frac d {\lsc}}$, by first considering special types of graph orientations, called \emph{threshold} orientations, and then carefully applying data-structure techniques. Notice that while in the uniformly bounded case we have $\degen(G_D) = O(\sqrt d)$, in the non-uniform case $\degen(G_D)$ could be much higher and so these new algorithms become a vast improvement. However these algorithms are mostly of interest from a theoretical perspective since in practice both the degeneracy and $\lsc$ are very small.

\section{\ThreeSUM{:} Conditional Lower Bounds}\label{s:conditional-lower-bound}
In this section we prove that conditioned on the \ThreeSUM{} conjecture we can prove lower bounds for the vertex-triangles problem, the ISG problem, and the (offline) unbounded DMOG problem. Since the other two versions of DMOG (uniformly and non-uniformly bounded) can solve the unbounded DMOG version, the lower bounds hold for those problems as well.
We use the following theorem proven by Kopelowitz, Pettie, and Porat~\cite{KPP14a}.

\begin{theorem}[\cite{KPP14a}]\label{thm:triangle_enumeration_lowerbound}
Assume \ThreeSUM{} requires $\Omega(n^{2-o(1)})$ expected time.
Then for any constant $0<x<1/2$, any algorithm for enumerating all triangles in a graph $G$ with $m$ edges, $\Theta(m^{1-x})$ vertices, and $\hat{d} = \degen(G) = \Theta(m^x)$, where $\hat{d}$ is the average degree of a vertex in $G$, must spend $\Omega(m\cdot \degen(G)^{1-o(1)})$ expected time.
\end{theorem}

\begin{proof}[Proof of Theorem~\ref{thm:vertex_triangles_LB}]
We reduce the triangle enumeration problem considered in Theorem~\ref{thm:triangle_enumeration_lowerbound} to the vertex-triangles problem. We preprocessing $G$ and then answer vertex-triangles queries on each of the $m^{1-x}$ vertices thereby enumerating all of the triangles in $G$. If we assume a sub-polynomial reporting time, then by Theorem~\ref{thm:triangle_enumeration_lowerbound} either the preprocessing takes $\Omega(m\cdot \degen(G)^{1-o(1)})$ time or each query must cost at least $\Omega(\frac{m\cdot \degen(G)^{1-o(1)}}{m^{1-x}}) = \Omega((m^{x}\degen(G))^{1 -o(1)} ) = \Omega((\hat d \cdot \degen(G))^{1-o(1)})$ time.
\end{proof}

\begin{proof}[Proof of Theorem~\ref{thm:ISG_lb} and Theorem~\ref{thm:DMOG_lb}]
We reduce the vertex-triangles problem considered in Theorem~\ref{thm:vertex_triangles_LB} to ISG as follows. We preprocess the graph $G$ for ISG queries. Now, when we want to answer a vertex-triangle query on some vertex $u$, we input all of the neighbors of $u$ into the ISG algorithm. Thus, there is a one-to-one correspondence between the edges reported by the ISG algorithm and the triangles in the output of the vertex-triangles query. Since each vertex-triangle query must cost $\Omega(\hat d \cdot \degen(G)^{1-o(1)})$ amortized expected time then the amortized expected time spent for each of the $\hat{d}$ neighbors of $u$ must be at least $\Omega(\degen(G)^{1-o(1)})$ amortized expected time. Since ISG is a special case of DMOG, and given Theorem~\ref{thm:ISG_lb}, the proof of Theorem~\ref{thm:DMOG_lb} follows directly.
\end{proof}

\begin{proof}[Proof of Theorem~\ref{thm:bounded_DMOG_lb}]
The proof is similar to the proofs of Theorems~~\ref{thm:ISG_lb} and~\ref{thm:DMOG_lb}.
First, we convert the input graph $G$ of the vertex-triangles problem
to a tripartite graph $G_T$ by creating three copies of the vertices $V_1,V_2,V_3$ and for each edge $(u,v)$ in $G$ we add $6$ edges to $G_T$ between all possible copies of $u$ and $v$. We also add a dummy vertex to $G_T$ with degree 0.
Each triangle in $G$ corresponds to a constant number of triangles in $G_T$. Let $\alpha$ be any positive integer and let $\beta = \alpha + 2\hat{d}$. We use ISG to solve vertex-triangles queries in Theorem~\ref{thm:ISG_lb}, but we only ask queries on the neighbors of vertices in $V_1$ in a specially tailored way as follows. We first list the neighbors of $u$ from $V_2$, followed by $\alpha$ copies of the dummy vertex, and then list the neighbors from $V_3$. From the construction of the tripartite graph and the input to the ISG algorithm, two vertices of an edge that is part of the output of the ISG algorithm must be separated in the input list by at least $\alpha$ vertices, and by at most the length of the list which is $\beta$. Thus, the time spent on each vertex must be at least $\Omega(\degen(G)^{1-o(1)}) = \Omega((m^x)^{1-o(1)}) = \Omega((\beta-\alpha)^{1-o(1)})$ amortized expected time. Continuing the reduction to the bounded DMOG problem completes the proof.
\end{proof}

\section{The Induced Subgraph Problem}\label{s:induced}
\noindent \paragraph {\bf An Upper Bound via Graph Orientations.} We make use of graph orientations, where the goal is to \emph{orient} the graph edges while providing some guarantee on the out-degrees of the vertices. Formally, an orientation of an undirected graph $G = (V,E)$ is called a {\em $c$-orientation} if every vertex has out-degree at most $c\ge1$. The notion of graph \emph{degeneracy} is closely related to graph orientations~\cite{AYZ95}. There is a simple linear time greedy algorithm by Chiba and Nishizeki~\cite{CN85} that assigns a $\degen(G)$-orientation of $G$. We use graph orientations for solving ISG problem as follows. First, we view a $c$-orientation as assigning ``responsibility'' for all data transfers occurring on an edge to one of its endpoints, depending on the direction of the edge in the orientation (regardless of the actual direction of the edge in the input graph $G$). We exploit this distinction by using the notation of an edge $e=(u,v)$ as oriented from $u$ to $v$, while $e$ could be directed either from $u$ to $v$ or from $v$ to $u$. We say that $u$ is \emph{responsible} for $e$, and that $e$ is \emph{assigned} to $u$. Furthermore, $u$ is a \emph{responsible-neighbor} of $v$ and $v$ is an \emph{assigned-neighbor} of $u$. Notice that in a $c$-orientation the number of assigned-neighbors of any vertex is at most $c$, while the number of responsible-neighbors could be much larger than $c$.

\paragraph{The Bipartite Graph.} We begin by converting $G=(V,E)$ to a bipartite graph by creating two copies of $V$ called $L$ (the left vertices) and $R$ (the right vertices). For every edge $(u,v)\in E$ we add an edge in the bipartite graph from $u_L\in L$ to $v_R\in R$, where $u_L$ is a copy of $u$ and $v_R$ is a copy of $v$ . All of the edges are directed from $L$ to $R$. Furthermore, each vertex in $V$ that arrives during query time is replaced by its two copies, first the copy from $R$ and then the copy from $L$. This ordering guarantees that a self loop in $G$ is not mistakenly reported the first time its single vertex arrives. Notice that the degeneracy of $G$ is unchanged, up to constant factors, due to this reduction. From here onwards we assume that $G$ is already in this bipartite representation.

\paragraph{The Data Structure.} We preprocess $G$ using the algorithm of~\cite{CN85}, thereby obtaining a $c$-orientation with $c=\degen(G)$. Each vertex $v\in R$ maintains a \emph{reporting list} $\rl_v$, which is a linked list containing its responsible-neighbors that have already appeared during the current query. When a vertex $v\in R$ arrives during query time, the elements in $\rl_v$ correspond to (some of the) edges that should be reported. Thus, reporting each such edge takes constant time via a linear scanning of $\rl_v$. The remaining edges to be reported are all assigned to $u$, so scanning the edges assigned to $u$ in $O(\degen(G))$ time suffices for listing them. When a vertex $u\in L$ arrives, $u$ is marked as arrived and we scan its assigned-neighbors, adding $u$ to their reporting lists. This also takes $O(\degen(G))$ time. Thus, we have proven Theorem~\ref{thm:induced_subgraph_algo}.
\begin{theorem}\label{thm:induced_subgraph_algo}
The ISG problem on a graph $G$ with $m$ edges and $n$ vertices can be solved online with $O(m+n)$ preprocessing time, $O(\degen(G) + op)$ time per query vertex, where $op$ is the number of edges reported at vertex arrival, and $O(m)$ space.
\end{theorem}

\subsection{Uniformly Bounded Edge Occurrences}\label{ss:uniformly-bounded-ISG}
In this case, the ISG problem is restricted with two positive integer parameters $\alpha$ and $\beta$ so that an edge $(v_j,v_i)$ can only be reported if $\alpha\leq i-j \leq\beta $. We still assume that the graph is bipartite, as before. The time window between $\beta$ time units ago and $\alpha$ time units ago is called the \emph{active window}. A reporting list $\rl_v$ for a vertex $v\in R$ contains the responsible-neighbors of $v$ which have appeared during the active window, without repetition. Since each responsible neighbor appears only once, the space consumption of all of these lists is $O(m)$. Furthermore, each vertex $u\in L$ maintains an ordered list of time stamps $\tau_u$ of the times $u$ appeared in the current active window.

Notice that the arrival of a vertex from $R$ should be treated immediately, while the treatment of the arrival of a vertex from $L$ should be delayed by $\alpha$ time units. Second, an algorithmic mechanism is needed for ``forgetting'' the arrival of a vertex after $\beta$ time units. These are addressed by maintaining a list of the last $\beta$ vertices that have arrived during query time. The complete proof of Theorem~\ref{thm:induced_subgraph_unifrom_bound_algo} is given in Appendix~\ref{app:induced}.

\begin{theorem}\label{thm:induced_subgraph_unifrom_bound_algo}
The Induced Subgraph problem with uniformly bounded edge occurrences on a graph $G$ with $m$ edges and $n$ vertices can be solved with $O(m+n)$ preprocessing time, $O(\degen(G) + op)$ time per query vertex, where $op$ is the number of edges reported at vertex arrival, and $O(m+ \beta)$ space.
\end{theorem}

\subsection{Non-Uniformly Bounded Edge Occurrences}\label{ss:non-uniformly-bounded-ISG}
We now consider the generalization to non-uniformly bounded edge occurrences, where each edge $e=(v_j,v_i)$ has its own boundaries $[\alpha_e, \beta_e]$ and can only be reported if $\alpha_e\leq i-j \leq\beta_e$. Notice that in this case the input graph is a multi-graph. The active window for this ISG version is the time window between $\beta^* = \max_{e\in E} \{\beta_e\}$ and $\alpha^* = \min_{e\in E} \{\alpha_e\}$ time units ago.

Similar to Section~\ref{ss:uniformly-bounded-ISG}, a list of the last $\beta^*$ vertices that have appeared is maintained so that a vertex can be removed from our data structure after $\beta^*$ units of time.

\section{DMOG via Graph Orientation}\label{s:DMOGinduced}

When extending ISG to online DMOG, the longer subpatterns introduce new challenges that need to be addressed. It is helpful to still consider the bipartite graph presentation of the DMOG instance, where vertices correspond to subpatterns and edges correspond to patterns. We use the algorithms from Section~\ref{s:induced} as basic building blocks in our algorithms for DMOG by treating a subpattern arriving as the vertex arriving in the appropriate graph. However, we now need to address the difficulties that arise from subpatterns being arbitrarily long strings.

To start off, we need a mechanism for determining when a subpattern arrives. One way of doing this is by using the the Aho-Corasick (AC) Automaton~\cite{AC75}, using a standard binary encoding technique so that each character costs $O(\log |\Sigma|)$ worst-case time. For simplicity we assume that $|\Sigma|$ is constant.
However, while in the ISG problem each character corresponds to the arrival of at most one subpattern, in the DMOG with unbounded gaps each arriving character may correspond to several subpatterns which all arrive at once, since a subpattern could be a proper suffix of another subpattern.
To address this issue we phrase the complexities of our algorithms in terms of $\lsc$ which is the maximum number of vertices in the bipartite graph that arrive due to a character arrival. This induces a multiplicative overhead of at most $\lsc$ in the query time per text character relative to the time used by the ISG algorithms.

Finally, there is an issue arising from subpatterns no longer being of length one, which for simplicity we first discuss this in the unbounded case. When $u\in L$ arrives and it has an assigned vertex $v\in R$ where $m_v$ is the length of the subpattern associated with $v$, then we do not want to report the edge $(u,v)$ until at least $m_v-1$ time units have passed, since the appearance of the subpattern of $v$ should not overlap with the appearance of the subpattern of $u$. Similarly, in the bounded case, we must delay the removal of $u$ from $\rl_v$ by at least $m_v-1$ time units. Notice that if we remove $u$ from $\rl_{v}$ after a delay of $m_v-1$, then we may be forced to remove a large number of such vertices at a given time, which may be unaffordable within our time bounds.
Our solution is to delay the removal of $u$ by $M-1$ time units, where $M$ is the length of the longest subpattern that corresponds to a vertex in $R$. This solves the issue of synchronization, however now there is a problem of some of the reporting lists being elements that should not be reported. Nevertheless, it is straightforward to spend time in a reporting list that corresponds to the size of the output using standard list and pointer techniques.
Combining these ideas with the algorithms in Section~\ref{s:induced} gives Theorems~\ref{thm:DMOG_unbounded_algo},~\ref{thm:DMOG_nomult_uniform_bound_algo} and~\ref{thm:DMOG_non_uniform_bound_algo}.

\begin{theorem}\label{thm:DMOG_unbounded_algo}
The DMOG problem with one gap and unbounded gap borders can be solved with $O(|D|)$ preprocessing time, $O(\degen(G_D)\cdot \lsc + op)$ time per text character, where $op$ is the number of patterns that are reported due to the character arriving, and $O(|D|)$ space.
\end{theorem}

\begin{theorem}\label{thm:DMOG_nomult_uniform_bound_algo}
The DMOG problem with uniformly bounded gap borders can be solved such that dictionary patterns are reported online in:
$O(|D|)$ preprocessing time, $O(\degen(G_D)\cdot \lsc+ op)$ time per text character, where $op$ is the number of patterns that are reported due to the character arriving, and $O(|D|+\lsc\cdot(\beta-\alpha+ M) + \alpha)$ space.
\end{theorem}

\begin{theorem}\label{thm:DMOG_non_uniform_bound_algo}
The DMOG problem with non-uniformly bounded gap borders can be solved such that dictionary patterns are reported online in:
$O(|D|)$ preprocessing time, $\ot(\degen(G_D)\cdot \lsc+ op)$ time per text character, where $op$ is the number of patterns that are reported due to the character arriving, and $\ot(|D|+\lsc\cdot \degen(G_D)(\beta^*-\alpha^*+ M)+\alpha^* )$ space.
\end{theorem}

\section{DMOG via Threshold Orientations}\label{s:threshold}

Sections~\ref{s:induced} and~\ref{s:DMOGinduced} focus on orientations whose out-degree is bounded by $\degen(G_D)$. Thus, when $\degen(G_D)=\sqrt d$ the DMOG algorithms take $O(\lsc \cdot \sqrt d)$ time. This is exacerbated in the non-uniform case where the degeneracy can be much larger, since the input graph is a multi-graph. In this section we show how in such dense inputs we can reduce the dependence on $\lsc\cdot \degen(G_D)$ to $\sqrt {\lsc\cdot d}$, by using a different method for orienting the graph which we refer to as a \emph{threshold} orientation.

\begin{definition}\label{d:frequent}
A vertex in $G_D$ is \emph{heavy} if it has more than $\sqrt{d/\lsc} $ neighbors, and \emph{light} otherwise.
\end{definition}

Our algorithms use two key properties. The first is that light vertices have at most $\sqrt{d/\lsc} $ neighbors, and the second is that the number of heavy vertices is less than $\sqrt{d/\lsc}$. As in the previous algorithms, for each vertex $v\in R$ we maintain either a reporting list $\rl_v$ in the uniform case or the data structure $S_v$ in the non-uniform case, and for each vertex $u\in L$ we maintain the list $\tau_u$. These structures enable dealing with edges where at least one of its endpoints is light, and so we orient all edges that touch a light vertex to leave that vertex, breaking ties arbitrarily if both vertices are light. In particular, if a vertex (heavy or light) $u\in L$ arrives then it adds the arrival to $\tau_u$, and if $u$ is light then it updates the data-structures in all of its out-going neighbors in $R$. Since there are at most $\lsc$ vertices arriving at a time, this costs at most $\ot(\lsc + \sqrt{\lsc\cdot d})$ time. If a vertex (heavy or light) $v\in R$ arrives then it reports relevant information from its responsible-neighbors, which must all be light vertices, in time proportional to the output size, and if $v$ is light then it can afford to scan all of its assigned-neighbors. Since there are at most $\lsc$ vertices arriving at a time, this costs at most $\ot(\lsc + \sqrt{\lsc\cdot d} + op)$ time.

The remaining task is reporting edges between two heavy vertices. From a very high level, we will leverage the fact that the number of heavy vertices is less than $\sqrt{\lsc \cdot d}$, and so even if the number of vertices from $L$ that arrive at the same time can be as large as $\lsc$ and the number of neighbors of each such vertex can be very large, the number of vertices in $R$ is still less than $\sqrt{\lsc \cdot d}$. So using a batched scan on all of $R$ will keep the time cost low. We show that after some preprocessing such a scan can produce the desired result.

\paragraph{The tree structure.}
Out treatment of heavy vertices uses a special tree-like structure $T$ among the subpatterns associated with vertices from $L$. In particular, the vertices of $T$ are the $O(\sqrt{\lsc\cdot d})$ vertices from $L$, where a vertex $u$ is an ancestor of a vertex $v$ if and only if the subpattern associated with $u$ is a suffix of the subpattern associated with $v$. To make the structure an actual tree we also add an additional special vertex corresponding to the empty string as the root of $T$ since it is a suffix of every subpattern. The tree $T$ can be constructed in linear time from the AC automaton of $D$. Notice that the depth of the tree is $\lsc$. Furthermore, the graph vertices arriving due to a text character arrival correspond to all of the vertices on some path from the root of $T$ to some vertex $u$, not including the root. This means that the arrival of $u$ implies the arrival of all of its at most $\lsc$ ancestors, not including the root. We emphasize that the AC automaton mechanics allow one to report only one vertex $u$ such that $u$ and all of its ancestors exactly correspond to the subpatterns that have arrived (implicitly).

We briefly overview the algorithm for the uniformly bounded case. The algorithm for the non-uniformly bounded case have a similar flavor.
Let $k=|R|$ and let $R=\{v_1,v_2,\ldots,v_{k}\}$. For each vertex $u\in L$ let $A_u$ be an array of size $k$ where the $i$'th location in $A_u$, denoted by $A_u[i]$, is a pointer to a list of all of the edges that connect an ancestor of $u$ in $T$ (which may be $u$) to $v_i$. When $u\in L$ arrives then for each pointer $A_u[i]\neq null$ we would like to add $A_u[i]$ to the beginning of $\rl_{v_i}$. Recall that if $A_u[i]$ was already in the reporting list $\rl_{v_i}$ then we remove the other older copy of $A_u[i]$ from $\rl_{v_i}$, in order to keep the space usage linear (as in Section~\ref{s:induced}). When a vertex $v_i\in R$ arrives we scan $\rl_{v_i}$ as in Section~\ref{s:DMOGinduced}, and for each edge $(u,v_i)$ in a $\rl_{v_i}$ that has appeared in the appropriate time frame (since we need to skip the tail of $\rl_v$ as in Section~\ref{s:DMOGinduced}) we scan the list $\tau_u$ (which is still maintained as before) and report all of the appropriate edges. Finally, after $\beta-\alpha+M+1$ time units have passed since the last time $u\in L$ has updated some reporting lists, we remove all of the $A_u[i]\neq null$ from each $\rl_{v_i}$. The rest of the implementation details discuss how the array $A_u$ can be computed efficiently during the query phase, with low preprocessing time and space.

The following Theorems are proven in Appendix~\ref{app:threshold}.

\begin{theorem}\label{t:uniform}
The DMOG problem with one gap and uniform gap borders can be solved with $O(|D|)$ preprocessing time, $O(\lsc+\sqrt{\lsc\cdot d }+ op)$ time per text character, and $O(|D|+ \lsc (\beta-\alpha +M)+\alpha)$ space.
\end{theorem}

\begin{theorem}\label{t:non-uniform}
The DMOG problem with one gap and non-uniform gap borders can be solved with $O(|D| + d(\beta^*-\alpha^*))$ preprocessing time, $\ot(\lsc + \sqrt{\lsc \cdot d}(\beta^*-\alpha^* + M) + op)$ time per query text character, and $\ot(|D| + d(\beta^*-\alpha^*)  + \sqrt{\lsc\cdot d}(\beta^*-\alpha^*+M)
+ \alpha^*)$ space.
\end{theorem}

\bibliographystyle{plain}
\bibliography{tsvi}

\newpage
\appendix
\section*{Appendix}
\section{Triangle Reporting Problems and Other Popular Conjectures.}\label{s:app_other_clbs}
As noted , many of the lower bounds conditioned on the \ThreeSUM{} conjecture are proven via a series of reductions that pass through either set disjointness queries or set intersection queries (see~\cite{Patrascu10,AW14,KPP14a}), as is the case with the triangle enumeration conditional lower bounds, which as stated lead to lower bounds for both edge-triangle queries and vertex-triangles queries. Since edge-triangle queries can be used to solve set disjointness, another natural candidate for a conjecture with which we can prove that edge-triangle queries are hard is the Boolean Matrix Multiplication (BMM) conjecture, which states that no $O(n^{3-\Omega(1)})$ combinatorial algorithm\footnote{There is no clear definition of a \emph{combinatorial algorithm}, and the notion that is accepted by the algorithmic community is that the way to establish if an algorithm is combinatorial or not is done by just looking at it.} exists for BMM on two $n\times n$ matrices. This is because the answer to a set disjointness query corresponds to the inner product of the characteristic vectors of the two sets, and each entry in the output of BMM is the inner product of one row and one vector from the input matrices. Notice that this approach derives a conditional lower bound from the BMM conjecture for the decision version of edge-triangle queries (does there exist a triangle containing the query edge), which can be solved via the reporting version considered above if we stop the query process after the first triangle is reported, thereby obtaining a conditional lower bound for the reporting version itself.

Just like the decision version of edge-triangle queries corresponds to the inner product of two boolean vectors, we can show that vertex-triangles queries correspond to the outer product of two vectors. However, outer products are too weak of a tool for proving conditional lower bounds from BMM, since the output of the outer product of two vectors of length $n$ is $n^2$, and in order to solve BMM using outer products we need to consider $n$ pairs of vectors and their outer products, resulting in $\Omega(n^3)$ information which is already too much.

Another candidate for proving the hardness of edge-triangle queries is the recent Online Matrix Vector (OMV) multiplication conjecture which states that there is no $O(n^{3-\Omega(1)})$ algorithm for multiplying an $n\times n$ matrix with $n$ vectors of length $n$ each, where the vectors arrive online and the output of the $i$th multiplication must be given prior to the arrival of the $(i+1)$th vector. Since multiplying a matrix with a vector corresponds can be solved via inner products, the connection to edge-triangle queries is clear. However, it is not clear how to use the OMV conjecture to prove some hardness on outer products, and so it is not clear if this conjecture can be used to prove the hardness of vertex-triangles queries, ISG, and DMOG.
\newpage
\section{Table of Upper Bounds}\label{app:upper-bound-table}

\begin{table}[!h]
\centering
\footnotesize{
  \begin{tabular}{|c|c|c|c|c|}
    \hline
    Gaps & Preprocessing & Query Time & Space \\
    Type & Time & per Text Character & \\
    \hline
    \hline
     unbounded & $O(|D|)$ & $O(\delta(G_D)\cdot\lsc +op)$ & $O(|D|)$ \\
    \hline
    uniform & & & \\
    bounds & $O(|D|)$ & $O(\delta(G_D)\cdot\lsc +op)$ & $O(|D|+\lsc(\beta-\alpha+M)+\alpha)$ \\
    \hline
    non-uniform & & & \\
    bounds & $O(|D|)$ & $\tilde{O}(\delta(G_D)\cdot\lsc +op)$& $\tilde{O}(|D|+\lsc(\beta^*-\alpha^*+M)+\alpha^*)$\\
    \hline
    uniform & & & \\
    bounds & $O(|D|)$ & $O(\lsc + \sqrt{\lsc\cdot d} +op)$ & $O(|D|+\lsc(\beta-\alpha+M)+\alpha)$ \\
    \hline
    non-uniform & $O(|D|+d(\beta^*-\alpha^*))$ & $\tilde{O}(\sqrt{\lsc\cdot d}(\beta^*-\alpha^*+M) +op)$& $\tilde{O}(|D|+d(\beta^*-\alpha^*) \makebox[70pt]{ }$ \\
    bounds & & & \,$\makebox[20pt]{ } +\sqrt{\lsc\cdot d}(\beta^*-\alpha^*+M)+\alpha^*)$\\
    \hline
  \end{tabular}\\
\caption{A summary of upper bounds for DMOG described in this paper.}
}
\end{table}

\section{Full Details for Section~\ref{s:induced}}\label{app:induced}

\begin{proof}[Proof of Theorem~\ref{thm:induced_subgraph_unifrom_bound_algo}]
The process for treating the arrival of $v_i$ is as follows. First, we remove $v_{i-\beta-1}$ from the list of the last $\beta$ vertices, and if $v_{i-\beta-1}\in L$ then we remove the time stamp of $i-\beta-1$ from $\tau_{v_{i-\beta-1}}$. If $\tau_{v_{i-\beta-1}}$ becomes empty then we remove $v_{i-\beta-1}$ from all of the reporting lists at the assigned-neighbors of $v_{i-\beta-1}$. Next, if $v_{i-\alpha}\in L$ then we treat it as if it just arrived. To do this, if $\tau_{v_{i-\alpha}}$ was empty before the current arrival, then we add $v_{i-\alpha}$ to the reporting lists of all of its assigned neighbors. We also add ${i-\alpha}$ to $\tau_{v_{i-\alpha}}$.

Finally, we are ready to treat $v_i$. First, $v_i$ is added to the list of the last $\beta$ vertices. Then, if $v_i\in R$ then we treat it as if it arrived right now by: (1) accessing the vertices in $\rl_{v_i}$ and reporting the edges corresponding to their time stamp lists, and (2) scanning the assigned-neighbors of $v_i$ to output any additional edges. The additional space usage for this part of the algorithm is another $O(\beta)$ words (for the list of the last $\beta$ vertices), and the time cost per vertex is $O(\degen(G)+op)$.
\end{proof}

\begin{proof}[Proof of Theorem~\ref{thm:induced_subgraph_mult_nonunifrom_bound_algo}]
Each vertex $u\in L$ maintains $\tau_v$ as in Section~\ref{ss:uniformly-bounded-ISG}. Since not all of the occurrences of $u$ are relevant when one of its responsible-neighbors from $R$ has arrived, the cost of filtering the list is at most $O(\log (\beta^*-\alpha^*)+k)$ where $k$ is the size of the output.

Similar to Section~\ref{ss:uniformly-bounded-ISG}, we maintain a list of the last $\beta^*$ vertices that have arrived during query time. The process for treating the arrival of $v_i$ is as follows. First we remove $v_{i-\beta^*-1}$ from the list of the last $\beta$ vertices, and if $v_{i-\beta-1}\in L$ then we remove the time stamp of $i-\beta^*-1$ from $\tau_{v_{i-\beta^*-1}}$. Next, we remove the at most $\degen(G)$ points that were created by $v_{i-\beta^*-1}$ from all of the orthogonal reporting data-structures at the assigned-neighbors of $v_{i-\beta^*-1}$. Notice that we only remove the points that were due to the arrival of $v_{i-\beta^*-1}$ at time $i-\beta^*-1$ (and not possibly other times).
Next, if $v_{i-\alpha^*}\in L$ then we add $i-\alpha^*$ to $\tau_{v_{i-\alpha^*}}$, and for each assigned-neighbor $v$ of $v_{i-\alpha^*}$ where $e=(v_{i-\alpha^*}, v)$ we insert the point $(i+\alpha_e, i+\beta_e)$ into $S_v$.

Finally, we are ready to treat $v_i$. First $v_i$ is added to the list of the last $\beta$ vertices. Then, if $v_i\in R$ then we: (1) execute a range reporting query in $S_v$ as described above (looking for all points whose first coordinate is at most $i$ and whose second coordinate is at least $i$) (2) and scan the at most $\degen(G)$ assigned-neighbors of $v_i$ to output any additional edges using the list of time stamps.

The additional space usage for the algorithm is $O(\beta^*)$ words for the list of the last $\beta^*$ vertices and all of the lists of time stamps, and another $O(\degen(G)(\beta^*-\alpha^*) \log^{7/8+\epsilon} (\degen(G)(\beta^*-\alpha^*)))$ words for all of the orthogonal range reporting data structures. The cost of processing a vertex in $L$ that arrived either $\alpha^*$ or $\beta^*+1$ time units ago is $O(\degen(G) \log^{7/8+\epsilon} (\degen(G)(\beta^*-\alpha^*)))$ time to add/remove points into/from orthogonal range reporting structures.
The cost of processing a vertex $v\in R$ that just arrived, in addition to the size of the current output, is $O(\frac{\log (\degen(G)(\beta^*-\alpha^*))}{ \log\log (\degen(G)(\beta^*-\alpha^*))})$ time to query $S_v$ and $O(\degen(G)\log (\beta^*-\alpha^*))$ time for scanning all of the assigned-neighbors of $v$.
\end{proof}

\section{Full Details for Section~\ref{s:threshold}}\label{app:threshold}

\begin{proof}[Proof of Theorem~\ref{t:uniform}]
Let $k=|R|$ and let $R=\{v_1,v_2,\ldots,v_{k}\}$. For each vertex $u\in L$ let $A_u$ be an array of size $k$ where the $i$'th location in $A_u$, denoted by $A_u[i]$, is a pointer to a list of all of the edges that connect an ancestor of $u$ in $T$ (which may be $u$) to $v_i$. When $u\in L$ arrives then for each pointer $A_u[i]\neq null$ we would like to add $A_u[i]$ to the beginning of $\rl_{v_i}$. Recall that if $A_u[i]$ was already in the reporting list $\rl_{v_i}$ then we remove the other older copy of $A_u[i]$ from $\rl_{v_i}$, in order to keep the space usage linear (as in Section~\ref{s:induced}). When a vertex $v_i\in R$ arrives we scan $\rl_{v_i}$ as in Section~\ref{s:DMOGinduced}, and for each edge $(u,v_i)$ in a $\rl_{v_i}$ that has appeared in the appropriate time frame (since we need to skip the tail of $\rl_v$ as in Section~\ref{s:DMOGinduced}) we scan the list $\tau_u$ (which is still maintained as before) and report all of the appropriate edges. Finally, after $\beta-\alpha+M+1$ time units have passed since the last time $u\in L$ has updated some reporting lists, we remove all of the $A_u[i]\neq null$ from each $\rl_{v_i}$. The rest of the implementation details discuss how the array $A_u$ can be computed efficiently during the query phase, with low preprocessing time and space.

\paragraph{Constructing arrays.}
We begin by leveraging the tree structure of $T$ where a vertex that arrives implies that all of its ancestors arrived as well.
For each vertex $u\in L$ and for each vertex $v_i\in R$ such that $e=(u,v_i)\in E_D$ we maintain a pointer $next(e)$ to an edge $e'=(u',v_i)$ where $u'$ is the lowest proper ancestor of $u$ in $T$ such that there is an edge from $u'$ to $v_i$. If no such vertex $u'$ exists then $next(e) = null$. It is straightforward to add these pointers in linear time, and their space usage is linear. Thus, it is enough for $A_u[i]$ to be a pointer to $e$ and the list of all of the ancestors of $u$ in $T$ that have edges touching $v_i$ is obtained through the $next(\cdot)$ pointers. Similarly, if for some $i'$ there is no edge $(u,v_{i'})$ then the entry of $A_u[i']$ should point to the edge $(u',v_{i'})$ where $u'$ is the lowest proper ancestor of $u$ in $T$ such that there is an edge from $u'$ to $v_i$. If no such edge exists then $A_u[i']=null$. In order to quickly find all of the entries of $A_u$ when $u$ arrives we will leverage a preprocessing phase.

Let $w:L\rightarrow [d]$ be a \emph{weight} function such that for $u\in L$ the weight $w(u)$ is the degree of $u$ in $G_D$. Notice that the weight of any vertex is at least $\sqrt{\frac d{\lsc}}$ since all of the vertices are assumed to be heavy.

The construction of these arrays will make use of the following procedure.
Let $u$ and $v$ be two vertices in $T$ where $v$ is a proper ancestor of $u$, and assume that $A_v$ has already been constructed.
In order to construct $A_u$ we first initialize all of its entries to $null$. Next we traverse the vertices from $u$ to $v$ by order, and for each vertex $w$ that we encounter on this path, for each edge $(w,v_i)$, if $A_u[i]=null$ then we set $A_u[i]$ to be a pointer to $(w,v_i)$. Once we reach $v$ we fill in the $null$ entries in $A_u$ with their corresponding entries in $A_v$ (some of which may also be $null$). The total time cost of this process is the total number of edges of all of the vertices on this path (notice that every vertex has weight at least 1) which is the sum of the weights of vertices in this path, and another $O(\sqrt{\lsc \cdot d})$ time for initializing $A_u$ and scanning $A_v$.

In order to utilize the construction procedure that was just described, we will maintain arrays only for specially chosen $O(\sqrt{\frac d{\lsc}})$ vertices in $T$, such that whenever we need to construct an array $A_u$ for some vertex $u$ during the query phase where $u$ is not special, the total weight of vertices on the path from $u$ to its closest special ancestor will be $O(\sqrt{\lsc \cdot d})$, so the total time cost for constructing $A_u$ will be $O(\sqrt{\lsc \cdot d})$.

\paragraph{Choosing special vertices.}
We partition $T$ into $O(\sqrt{\frac d{\lsc}})$ small subtrees such that each subtree has total weight $\Omega(\sqrt{\lsc \cdot d})$, except for possibly the subtree containing the root of $T$. This guarantees that the total number of small subtrees is $O(\sqrt{\frac{d}{\lsc}})$. The special vertices are the roots of these small subtrees.

The partitioning is obtained by (greedily) peeling small subtrees in the bottom of $T$. Specifically, let $T_u$ be the subtree of $T$ rooted at $u$ and let $w(T_u)$ be the total weight of vertices in $T_u$. Then we iteratively peel a subtree $T_u$ such that $w(T_u)$ is at least $\sqrt{\lsc \cdot d}$ but the total weight of each subtree of a child of $u$ in $T$ is (separately) less than $\sqrt{\lsc \cdot d}$.
This peeling continues until no such subtree exists, in which case the remaining subtree must have total weight less than $\sqrt{\lsc \cdot d}$
and it is the last small subtree (it also contains the root). It is straightforward to implement the partitioning in linear time using a post-order traversal.
Notice that by the construction method, for any vertex $u$ that is not the root of $T$, the total weight of vertices on the path from $u$ to its closest proper ancestor that is a special vertex is $O(\sqrt{\lsc \cdot d})$.

To compute the arrays for all of the special vertices we use a top-down approach by first constructing the array for the root of $T$ and then we compute each $A_u$ for a special vertex $u$ only after the array for the closest proper special ancestor $v$ of $u$ was constructed. Using the construction procedure above and the property that the total weights of vertices on the path from $u$ to $v$ is $O(\sqrt{\lsc \cdot d})$ the time to construct the array for each special vertex is $O(\sqrt{\lsc \cdot d})$. Since the number of special vertices is $O(\sqrt{\frac d {\lsc}})$ the total preprocessing phase costs $O(d)$ time. Similarly, the process of constructing the array $A_u$ for a vertex $u\in L$ that arrived during a query costs $O(\sqrt{\lsc \cdot d})$ time. Notice that if $u$ is not a special vertex then we construct $A_u$ temporarily when $u$ arrives, use $A_u$ in order to update the reporting lists of vertices in $R$, and then delete $A_u$. Moreover, if $\beta-\alpha+M+1$ time units have passed since the last time $u$ has updated reporting lists, we again construct $A_u$ only temporarily in order to delete the information that should be removed.
\end{proof}

\begin{proof}[Proof of Theorem~\ref{t:non-uniform}]
As in the uniform case let $R=\{v_1,v_2,\ldots,v_{k}\}$.
Each vertex $v_i\in R$ will maintain a cyclic \emph{active window} array $AW_i$ of size $\beta^*-\alpha^*+M+1$ where the $j$'th location in $AW_i$, denoted by $AW_i[j]$, will be a pointer to a list of lists of edges that all need to be reported if $v_i$ will appear in $j-1$ time units from now. Since $k=O(\sqrt{\lsc\cdot d})$, every time a vertex arrives during the query phase we shift all of the active window arrays by one position. This shift is implemented by incrementing the starting position of each active window array in a cyclic manner, costing $O(1)$ time per vertex in $R$.
When a vertex $v_i$ appears, all of the edges pointed to by location $AW_i[1]$ will be reported, in time proportional to the size of the output. In the remainder of this section we focus on how to maintain the active array windows correctly as vertices from $L$ appear during the query phase.

We again leverage the tree structure of $T$ where a vertex that arrives implies that all of its ancestors arrived as well.
Recall that for an edge $e$ its gap boundaries are denoted by $\alpha_e$ and $\beta_e$.
For each vertex $u\in L$ and for each vertex $v_i\in R$ such that $e=(u,v_i)\in E_D$ we maintain an array $next_e$ of size $\beta_e-\alpha_e+1$. It will be helpful to treat the indices of the array as starting from $\alpha_e$ and ending at $\beta_e$.
So for $\alpha_e\leq j \leq \beta_e$, the $j$'th location in $next_e$, denoted by $next_e[j]$, points to an edge $e'=(u',v_i)$ where $u'$ is the lowest ancestor of $u$ in $T$ (possibly $u$ itself) such that:
\begin{itemize}
\item{} There is an edge $e' = (u',v_i)$.
\item{} $\alpha_{e'}\leq j\leq \beta_{e'}$.
\item{} It is not possible to reach $e$ by iteratively following $next_r[j]$.
\end{itemize}
If no such edge exists then $next_e[j] = null$. The total space usage for all of the $next$ pointer arrays is $\rho:=  \sum_{e\in E_D} (\beta_e - \alpha_e+1) \leq d(\beta^*-\alpha^*)$. Also, it is straightforward to construct these arrays in $O(\rho)$ time. Notice that from the way we constructed the $next$ pointer arrays, for a given vertex $u\in L$ all of the multi-edges $(u,v_i)$ whose boundaries contain $j$ form a contiguous list when considering the set of pointers $next_{e}[j]$ for each such edge $e$. We denote the first edge on this contiguous list by $head_{u,i,j}$. Let the weight of a node $u\in L$, denoted by $w(u)$, be the number of pointers of the form $head_{u,i,j}\neq null$. Notice that the total weight of all of the vertices in $T$ is $O(d(\alpha^*-\beta^*))$ (as opposed to the uniform case where the total weight was $O(d)$).

For each vertex $u\in L$ and $v_i\in R$ let $W_{u,i}$ be an array of size $\beta^*-\alpha^*+1$. As before, it will be helpful to treat the indices of the array as starting from $\alpha^*$ and ending at $\beta^*$. The $j$'th location in $W_{u,i}$, denoted by $W_{u,i}[j]$, is a pointer to $head_{u,i,j}$ thereby giving access to a list of all of the edges in $E_d$ that: (1) touch $v_i$, (2) touch an ancestor of $u$ in $T$, and (3) their boundaries contain $j$. When $u\in L$ arrives then for each pointer $W_{u,i}[j]\neq null$ we would like to add $W_{u,i}[j]$ to the list at $AW_{i}[j + 1 +m_{v_i}]$ where $m_{v_i}$ is the length of the subpattern corresponding to $v_i$, since in $j+1+m_{v_i}$ time units from now, if $v_i\in R$ arrives we will want to report all of the edges pointed to by $W_{u,i}[j]$.

If we were to precompute all of the arrays $W_{u,i}$ the time and space would be $O(\lsc \cdot d (\beta^*-\alpha^*))$, which may be rather large. To reduce this cost we use techniques that we used in the uniform case, as follows.

The construction of our array will make use of the following procedure.
Let $u$ and $v$ be two vertices in $T$ where $v$ is a proper ancestor of $u$, and assume that we have already constructed $W_{v,i}$.
In order to construct $W_{u,i}$ we first initialize all of the entries to $null$. Next we traverse the vertices from $u$ to $v$ by order, and for each vertex $w$ that we encounter on this path, for each $head_{w,i,j}$, if $W_{u,i}[j]=null$ then we set $W_{u,i}[j]$ to be a pointer to $head_{w,i,j}$. Once we reach $v$ we fill in the $null$ entries in $W_{u,i}$ with their corresponding entries in $W_{v,i}$ (some of which may also be $null$). The total time cost of this process is the total weight of vertices on this path (notice that every vertex has at least 1 head pointer), and another $O(\sqrt{\lsc \cdot d}(\beta^*-\alpha^*))$ time for initializing $W_{u,i}$ and scanning $W_{v,i}$.

In order to utilize the construction procedure that was just described, we will maintain arrays only for specially chosen $O(\sqrt{\frac d{\lsc}})$ vertices in $T$, such that whenever we need to construct an array $A_u$ for some vertex $u$ during the query phase where $u$ is not special, the total weight of vertices on the path from $u$ to its closest special ancestor will be $O(\sqrt{\lsc \cdot d}(\beta^*-\alpha^*))$, so the total time cost for constructing $W_{u,i}$ will be $O(\sqrt{\lsc \cdot d}(\beta^*-\alpha^*))$.

\paragraph{Choosing special vertices.}
As in the uniform case, we partition $T$ into $O(\sqrt{\frac d{\lsc}})$ small subtrees such that each subtree has total weight $\Omega(\sqrt{\lsc \cdot d}(\beta^*-\alpha^*))$, except for possibly the subtree containing the root of $T$. This guarantees that the total number of small subtrees is $O(\sqrt{\frac{d}{\lsc}})$. The special vertices are the roots of these small subtrees.

The partitioning is obtained by (greedily) peeling small subtrees in the bottom of $T$, as in the uniform case, and costs linear time using a post-order traversal.
Also, as in the uniform case, the total weight of vertices on the path from $u$ to its closest proper ancestor that is a special vertex is $O(\sqrt{\lsc \cdot d}(\beta^*-\alpha^*))$. Computing $W_{v,i}$ for all special vertices $v$ and all $i$ is executed using a top-down approach, as in the uniform case, and the time to construct the array for each special vertex is $O(\sqrt{\lsc \cdot d}(\beta^*-\alpha^*))$. Since the number of special vertices is $O(\sqrt{\frac d {\lsc}})$ the total preprocessing phase costs $O(d(\beta^*-\alpha^*))$ time. Similarly, the process of constructing $W_{u,i}$ for a vertex $u\in L$ that arrived during a query costs $O(\sqrt{\lsc \cdot d}(\beta^*-\alpha^*))$ time. The rest of the details also follow the ideas from the uniform case: $W_{u,i}$ is temporarily constructed when $u$ arrives, and again after $\beta^*-\alpha^*+1 +M$ time in order to delete the information that should be removed.
\end{proof}

\end{document}